\documentclass[11pt]{amsart}
\usepackage{amsmath,amsthm,amssymb,amscd, amsfonts,hyperref,graphicx, cite}
\newtheorem{thm}{Theorem}[section]

\newtheorem{prop}[thm]{Proposition}

\newtheorem{defn}[thm]{Definition}

\newtheorem{prps}[thm]{Properties}
\newtheorem{rmk}[thm]{Remark}
\theoremstyle{plain}
\numberwithin{equation}{thm}
\numberwithin{equation}{section}
\theoremstyle{remark}

\def\phi{\varphi}

\title[Correlation Angles and Inner Products]{Correlation Angles and Inner Products: \\
Application to a problem from Physics}
\author{Adam Towsley}

\address{
Adam Towsley\\
Department of Mathematics\\
University of Rochester\\
Rochester, NY 14627 \\
}

\email{towsley@math.rochester.edu}

\author{Jonathan Pakianathan}

\address{
Jonathan Pakianathan\\
Department of Mathematics\\
University of Rochester\\
Rochester, NY 14627 \\
}

\email{jonpak@math.rochester.edu}

\author{David H. Douglass}

\address{
David H. Douglass\\
Department of Physics and Astronomy
University of Rochester\\
Rochester, NY 14627 \\
}

\email{douglass@pas.rochester.edu}

\begin{document}
\hfill (Accepted by ISRN Applied Mathematics. August 11, 2011.)\\\\
\begin{abstract}
Covariance is used as an inner product on a formal vector space built on $n$
random variables to define measures of correlation $M_d$ across a set of vectors in a $d$-dimensional space. For $d=1$, one has the diameter; for $d=2$, one has an area. These concepts are directly applied to correlation studies in climate science.

\end{abstract}

\maketitle

\section{Introduction}
In a study of the earth's climate system Douglass \cite{Douglass2010} considered the correlation among a set of $N$ climate indices. A distance $d$ between two indices $i$ and $j$ was defined as 
\begin{equation} d_{ij} \left( t \right) = \cos^{-1}\left( \vert \phi_{ij}\left( t \right) \vert \right) 
\end{equation}
where $\phi_{ij}$ is the Pearson correlation coefficient. It was stated that $d$ satisfies the conditions to be a metric. The measure of correlation, or closeness, among the $N$ indices was taken to be the diameter $D$. 
\begin{equation}\label{diam}
D_{I_0}\left( t \right) = \max \left\lbrace d_{ij}\left( t \right) \vert i,j \in I_0 \right\rbrace
\end{equation}
Equation \ref{diam} was applied to the data from a global set of four climate indices to determine the correlation among them (minimum in $D$) and to infer 18 changes in the state since 1970. (See section \ref{appendix}.) It was pointed out that the topological diameter $D$, as a measure of phase locking among the indices, is convenient for computation but was probably not the best measure. It was suggested that a better measure of correlation among the $N$ indices could be based upon the area of the spherical triangles created by the $N$ vectors on the unit sphere.

This paper gives a proof that $d_{ij}$ is a metric and generalizes the diameter to higher dimensions. In addition, the data of \cite{Douglass2010} are analyzed using this generalization to areas (see section \ref{appendix}) and many new abrupt climate changes are identified.

\section{Probability}
Let $X$ and $Y$ be random variables with expected values $E\left(X \right) = \mu$ and $E \left( Y \right) = \nu$. With these values we make several standard definitions.

\begin{defn}\label{Variance}
The Variance of $X$ is defined as $$Var \left[ X \right] = E \left[ \left( X - \mu \right)^2 \right]$$
\end{defn}

\begin{defn}\label{Covariance}
The Covariance of $X$ and $Y$ is defined as $$Covar\left[X,Y \right] = E \left[ \left(X- \mu \right) \left( Y - \nu \right) \right]$$
\end{defn}

We now list a few basic properties of variance and covariance (found in \cite{Ross}).

\begin{prps}\label{Basics}
For $X$ and $Y$ as above:
\begin{enumerate}
\item $Covar$ is symmetric.
\item $Covar$ is bilinear.
\item $Var \left[ - \right]$ is a quadratic form.
\item $Covar\left[X,Y \right] = E \left[ XY \right] - E \left[ X \right] E \left[ Y \right]$.
\item $Covar\left[ X, X \right] = Var \left[ X \right]$, the variance of $X$.

\end{enumerate}
\begin{proof}\
\begin{enumerate}
\item See \cite{Ross}, page 323.
\item Follows easily from the definition.
\item See \cite{Ross}, page 323.
\item See \cite{Ross}, page 323.
\end{enumerate}
\end{proof}
\end{prps}

\section{Vector Spaces}
The first way most students learn to compare two vectors is through the dot product. The dot product is one example of the more general idea of an inner product. Here we define an inner product and prove that covariance is an inner product.
\begin{defn}\label{Inner Product}
For any real vector space $V$ an inner product is a map $$ \langle - , - \rangle : V \times V \rightarrow \mathbb{R}$$ that satisfies the following properties for every $u,v,w \in V$ and $a \in \mathbb{R}$:

\begin{enumerate}
\item $\langle u+v, w \rangle = \langle u, w \rangle + \langle v, w \rangle$
\item $\langle av, w \rangle = a \langle v,w \rangle$
\item $\langle v,w \rangle = \langle w,v \rangle$
\item $\langle v, v \rangle \geq 0$ and \\ $\langle v,v \rangle = 0$ if and only if $ v = 0 $
\end{enumerate}
\end{defn}

We will now construct a vector space for which covariance is an inner product. Let $\left\lbrace X_1, X_2, \dots, X_n \right\rbrace$ be a set of $n$ random variables. Also let $V = Span_{\mathbb{R}}\left(X_1, X_2, \dots, X_n \right)$, the formal $\mathbb{R}$-vector space with basis elements $\left\lbrace X_1, X_2, \dots, X_n \right\rbrace$. We must put one mild hypothesis upon $V$ in order for it to have the desired properties. The hypothesis is that the vectors must be ``probabilistically independent". i.e. for any $c_1, \dots, c_n \in \mathbb{R}$, we have that $Var\left[ c_1X_1 + \cdots + c_n X_n \right] = 0$ if and only if $c_1 = \cdots = c_n = 0$. It should be noted that this independence is in no way related to the linear independence of the random variables.

\begin{prop}\label{IP}
Let $V = Span_\mathbb{R} \left( X_1, X_2, \dots, X_n\right)$, the formal $\mathbb{R}$-vector space generated by the random variables $\left\lbrace X_1, X_2, \dots, X_n \right\rbrace$ which are probabilistically independent, then covariance is an inner product on $V$.
\end{prop}
\begin{proof}
We must prove the four properties from definition \ref{Inner Product}.\\
i), ii) and iii) follow immediately from proposition \ref{Basics}.\\
iv) $Covar\left( X, X \right) = E \left[ \left( x- \mu \right)^2 \right] \geq 0$. The non-negativity is obvious as we are squaring a real number. The condition that $Covar\left( X, X \right) = 0 \Leftrightarrow X =0$ follows from the probablisitic independence of $\left\lbrace X_1, \dots, X_n \right\rbrace$. 
\end{proof}

The proposition implies that $V$ is an inner-product space (a vector space equipped with an inner-product), and as such it has a norm defined by $\Vert X \Vert = \sqrt{ Covar \left( X,X \right)} = SD \left( X \right)$, where $SD \left( X \right)$ is the standard deviation of $X$. Additionally it follows from the Cauchy-Schwartz inequality (\cite{HoffmanKunze}) that $\vert Covar \left( X,Y \right) \vert \leq SD \left( X \right) SD \left( Y \right)$. 

Using the inner product on $V$ we are able to define an angle between two vectors. To do this we first define a new map $\rho : \left( V \setminus \left\lbrace 0 \right\rbrace \right) \times \left( V \setminus \left\lbrace 0 \right\rbrace  \right) \rightarrow \mathbb{R} $ using the standard definition of correlation $$\rho\left(X,Y \right) = \dfrac{Covar \left( X,Y \right) }{SD \left( X \right) SD \left( Y \right)}$$ By the Cauchy-Schwartz inequality we can easily see that $ \left| \rho \left( X, Y \right) \right| \leq 1$, as such we implicitly define $\Gamma$, the angle between $X$ and $Y$, as follows: $$Covar\left( X,Y \right) = SD \left( X \right) SD \left( Y \right) \cos \left( \Gamma \right)$$ Therefore $\rho\left(X,Y \right) = \dfrac{Covar \left( X,Y \right) }{SD \left( X \right) SD \left( Y \right)} = \cos \left( \Gamma \right)$. 

\begin{defn}\label{corrangle}
$\Gamma \left(X,Y \right) = \cos^{-1} \left( \rho \left( X, Y \right) \right)$ is the ``Correlation Angle" of $X$ and $Y$.
\end{defn}
 
Our definition of $\Gamma$ is the standard method of defining an angle from the covariance (or any other) inner product. We will show that $\Gamma$ is a `metric' on the unit sphere of $V$.

\begin{defn}\label{metric}
For any set $S$ a map $d: S \times S \rightarrow \mathbb{R}$ is a metric if for any $x,y,z \in S$ the following properties are satisfied.
\begin{subequations}
\begin{align}
d \left( x,y \right) \geq 0 \text{ with }d \left( x,y \right) = 0 \Leftrightarrow x=y \text{ (positive definite)}\\
d \left( x,y \right) = d \left(y,x \right) \text{ (symmetry)}\\
d \left( x,z \right) \leq d \left( x,y\right) + d\left( y,z \right) \text{ (triangle inequality)}
\end{align}
\end{subequations}
\end{defn}

\begin{thm}\label{GammaMetric}
The map $\Gamma: V \times V \rightarrow \mathbb{R}$ from definition \ref{corrangle} is a metric on $S \left(V \right) = $ the unit sphere of $V$.
\end{thm}

\begin{proof}
We must prove that $\Gamma$ satisfies the 3 conditions in definition \ref{metric}.
\begin{enumerate}
\item[(a)] $\cos^{-1}: \left[-1,1\right] \rightarrow \left[0, \pi \right]$ so the non-negativity is satisfied trivially. It remains to show that $\Gamma \left( X,Y \right) = 0 \Leftrightarrow X=Y$. This true because if the angle between two vectors is zero, then they are (positive) scalar multiples of each other. Thus since $X$ and $Y$ are unit vectors, if $\Gamma \left(X, Y \right) = 0 $ we must have $X = Y$. 
\item[(b)] $\Gamma \left( X,Y \right) = \cos^{-1} \left( \rho \left( X,Y \right) \right) = \cos^{-1} \left( \dfrac{Covar \left( X,Y \right)}{SD \left( X \right) SD \left( Y \right) } \right)= $ \\ $\cos^{-1} \left( \dfrac{Covar \left( Y, X \right)}{ SD \left( Y \right) SD \left( X \right) }\right) = \cos^{-1} \left( \rho \left(Y,X \right) \right) = \Gamma \left(Y,X \right)$
\item[(c)] To prove the triangle inequality, a geometric idea in itself, we delve into the geometry being defined. We will complete this part of the proof in section \ref{Geometry}.
\end{enumerate}
\end{proof}

Our metric $\Gamma$ allows us to measure the correlation between two vectors.

\begin{defn}\label{Corr}
For $X$, $Y$, $\Gamma$ and $\rho$ as above:

\begin{enumerate}
\item If $\Gamma = 0$ $\left(\rho =1\right) $ then $X$ and $Y$ are maximally positively-correlated.
\item If $\Gamma = \pi $ $\left(\rho =-1\right) $ then they are maximally negatively-correlated.
\item If $\Gamma = \pi/2$ $\left(\rho =0\right) $ then $X$ and $Y$ are uncorrelated.
\end{enumerate}
\end{defn}

It should be noted that cases (i) and (ii) are both considered to be ``maximally correlated".  

\section{A Geometric Interpretation}\label{Geometry}
The vector space $V$ with inner product $Covar$ lends itself nicely to a geometric interpretation. First we must establish a small amount of background.

Consider $S$, the standard unit sphere in Euclidean $n$-space ($\mathbb{R}^n$). Great circles are the intersection of a plane through the origin and $S$. They share many properties with the standard idea of lines in Euclidean space, including the property that they define the shortest path between any two points. For a thorough treatment of great circles as lines on a sphere see \cite{Henderson}, \cite{HendersonDiff}, \cite{BON} or \cite{Morita}.

For any two non-zero vectors $v_1$ and $v_2$ in $\mathbb{R}^n$ let $\theta$ be the (minimal) angle formed by $v_1$ and $v_2$. The unit vectors $\hat{v}_1$ and $\hat{v}_2$, corresponding to $v_1$ and $v_2$, define two points $p_1$ and $p_2$ on $S$. In order to measure the distance from $p_1$ to $p_2$ along $S$ we take the length of the arc on great circle between the two points. By definition this is the radian measure of $\theta$.

If $V$, the vector space considered in section 3, is thought of as $\mathbb{R}^n$ with $v_1$ and $v_2$ any two vectors, then we can compute the spherical distance between $v_1$ and $v_2$, namely the distance between $p_1$ and $p_2$ on $S$. We call this quantity $\Gamma$: $$d_{\text{spherical}} \left( v_1, v_2 \right) = arccos \left( \rho \left( \hat{v}_1, \hat{v}_2 \right) \right) = \Gamma$$

Thus far we have identified the inner product space $ \left( V, Covar \right)$ as $\mathbb{R}^n$. We solidify this intuition with the following proposition. First we define $A= \left( A_{i,j } \right) = \left(Covar \left( X_i, X_j \right)\right)$, a real valued symmetric matrix. As in \cite{HoffmanKunze} we use $A$ to create the inner product on $\mathbb{R}^n$.

\begin{prop}\label{twisted dot product}
The inner product space $\left\langle Span_\mathbb{R} \left( X_1, \dots , X_n \right), Covar \right\rangle \cong \left\langle \mathbb{R}^n, \cdot_A \right\rangle$ w here $\cdot_A$ is a `twisted dot product' defined for two vectors $\left( c_1, \dots, c_n\right)$ and $\left( d_1, \dots , d_n \right)$ as  $$ \left( c_1, \dots, c_n \right) \cdot_A \left( d_1, \dots, d_n \right):= \left(c_1, \dots, c_n \right) A \left( \begin{array}{c}
d_1 \\ 
\vdots \\
d_n 
\end{array} \right)$$
\end{prop}

\begin{proof}This follows from the standard method of representing an inner-product by a matrix. (See \cite{HoffmanKunze} chapter 8.1).
\end{proof}

Now we return to our proof of \ref{GammaMetric}.
\begin{proof}[Proof of \ref{GammaMetric}(iii)]
Let $X,Y,Z \in V$ be unit vectors. We have left to show that $\Gamma\left( X,Y \right) + \Gamma \left(Y,Z \right) \geq \Gamma \left( X,Z \right)$. \\
Because $X$ and $Z$ are unit vectors, $\Gamma\left( X,Z\right)$ is the geodesic distance between $X$ and $Z$. Since geodesic distance satisfies the triangle inequality, $\Gamma$ must as well.
\end{proof}

\section{Projective Metric}

For scientists, $\rho = \pm 1$ (equivalently $\Gamma = 0 $ or $\Gamma = \pi$) are often both considered to be ``maximally correlated", for example see \cite{Douglass2010}. To take this into account we modify our metric on the unit sphere of $V$. We think of $V$ as a projective space, the space of lines through the origin of $V$. We denote this space as $\mathbb{P}\left( V \right)$.

Our original correlation angle $\Gamma$ is modified to be: 
$$\Gamma^{\prime} = \arccos \vert \rho \left( X, Y \right) \vert = 
	\begin{cases}
		\Gamma & : 0 \leq \Gamma \leq \pi/2\\
		\pi - \Gamma & : \pi/2 \leq \Gamma \leq \pi	
	\end{cases} $$
	
\begin{prop}\label{projective metric}
$\Gamma^\prime \left( X,Y \right)$ is a metric on $\mathbb{P}\left(\mathbb{R}^n\right)$.
\end{prop}
\begin{proof}
We must show that the three conditions of \ref{metric} are met.
\begin{enumerate}
\item $\Gamma^\prime = 0$ corresponds to a correlation angle of $0$ or $\pi$. The two vectors are either in the same direction or opposite direction. In either case they determine the same line through the origin and hence correspond to the same point in projective space.
\item As in \ref{metric} the symmetry of $\Gamma^\prime$ follows from the symmetry of $\rho$.
\item As before, the triangle inequality follows as $\Gamma^\prime$ is the geodesic distance for a projective space.
\end{enumerate}

\end{proof}

The metric $\Gamma^\prime \left( X,Y \right)$ gives the angular distance between $X$ and $Y$. If $\rho \left( X,Y \right) = \pm 1$ (what we called a ``maximal correlation") then $\Gamma^\prime = 0$ however if $\rho \left( X,Y \right) = 0$, which we called orthogonality or non-correlation, then $\Gamma^\prime = \dfrac{\pi}{2}$.  
\begin{prop}\label{projective space}
Let $\Gamma^\prime$ be the metric $\cos^{-1} \left( \rho^\prime \left( X,Y \right)\right) $, then the pair $\left(\mathbb{P}\left(V\right), \Gamma^\prime \right)$ is a projective metric space.
\end{prop}
\begin{proof}
This is by construction.
\end{proof}
\section{Time Dependence}\label{TD}
Until this point we have treated our random variables $\left\lbrace X_1, \dots, X_n \right\rbrace$ as being time-independent. However, random variables often depend on time. Therefore we will now consider each random variable as depending discretely on time. It should be noted that what follows is essentially a replication of what has come before, however $X$ and $Y$ are now treated as vectors instead of singleton points. Vectors, however, are just points of $V$. The additional theory and notation is simply a means of dealing with the additional information.

To make our $n$ random variables time dependent they will now be given as:

$$\begin{array}{c} X_1 = \left\lbrace X_1 \left( t \right) , X_1 \left( t+1 \right) , X_1 \left( t + 2 \right) , \dots \right\rbrace \\
X_2 = \left\lbrace X_2 \left( t \right) , X_2 \left( t+1 \right) , X_2 \left( t + 2 \right) , \dots \right\rbrace \\
\vdots \\
X_n = \left\lbrace X_n \left( t \right), X_n \left( t+1 \right), X_n \left( t + 2 \right), \dots \right\rbrace
\end{array} $$

We must now redefine the covariance. We do this by looking at a time window starting at time $t$ with a duration of $K$, where $K$ is called the summation window. 
$$Covar\left( X_i, X_j \right) = \frac{1}{K} \displaystyle\sum_{l=t }^{t+K-1} \left( X_i\left( l \right) - \mu \right) \left( X_j \left( l \right)  - \nu \right)$$ 
Where $\mu$ and $\nu$ are the sample means in the summation window of $X_i$ and $X_j$ respectively. i.e. $\mu = \dfrac{1}{K} \displaystyle\sum_{l=t}^{p=t+K-1}X_i \left( l \right)$. 

If we think of $\hat{X_i}$ and $\hat{X_j}$ as the vectors $\hat{X_i} = \left( X_i\left( t \right), \dots, X_i \left( t+K-1 \right) \right)$ (resp. for $\hat{X_j}$)  then we get that $$Covar \left(X_i, X_j \right) = \frac{1}{K} \left( \hat{X_i} - \hat{\mu} \right) \cdot \left( \hat{X_j} - \hat{\nu} \right)$$ where``$\cdot$" is the standard Euclidean dot product and $\hat{\mu}$ is the length $K$ vector $\left( \mu, \dots, \mu \right)$ (resp. for $\hat{\nu}$). This is called the ``Pearson Covariance".

In other words if we define the vectors $\hat{\hat{X_i}} = \dfrac{\hat{X_i}\left( t \right) - \mu}{\sqrt{K}}$ and $\hat{\hat{X_j}} = \dfrac{\hat{X_j}\left( t \right) - \mu}{\sqrt{K}}$ then we define the Pearson Correlation as follows.

\begin{defn}\label{Pearson}
$Covar_{Pearson}\left( X_i,X_j \right) = \hat{\hat{X_i}} \cdot \hat{\hat{X_j}}$, where ``$\cdot$" is the usual Euclidean inner product. 
\end{defn}

Now we define the Pearson Correlation  as  $$\hat{\rho} \left( X_i, X_j \right) = \dfrac{ Covar_{Pearson} \left( X_i,X_j \right) }{\sqrt{Covar_{Pearson} \left( X_i,X_i\right)Covar_{Pearson} \left( X_j,X_j \right) }} = \cos \left( \hat{\Gamma} \right)$$ Here again $\hat{\Gamma}$ corresponds to the standard Euclidean angle, known as the Pearson Correlation Angle, and the resulting metric is the standard metric studied in classical spherical geometry. (See \cite{Henderson}, \cite{HendersonDiff}, \cite{BON} or \cite{Morita}). 

\begin{rmk}
The angle $\hat{\Gamma}$ between $\hat{\hat{X_i}}$ and $\hat{\hat{X_j}}$ is the same as the angle between $\bar{X_i}$ and $\bar{X_j}$, the unit vectors corresponding to $\hat{\hat{X_i}}$ and $\hat{\hat{X_j}}$.
\end{rmk}
\section{Correlation Measures: $M_n$ and $M_{n,a}$}\label{CorrMeasures}

To this point we have developed a method that will numerically tell us the correlation between two vectors. In this section we will create two sets of functions that allow us to measure the correlation across a set of vectors. The first set, $\left\lbrace M_{i,a} \right\rbrace$ is based upon taking the volumes of $i$-simplices (a 1-simplex is a line, a 2-simplex a triangle, a 3-simplex is a tetrahedron, etc.) The set of $M_{i,a}$ benefits from computability, but is not as precise as the second set of measures $\left\lbrace M_i \right\rbrace$, that measure the volume of $i$-dimensional convex hulls. 

Given a set of vectors $\left\lbrace X_1, \dots, X_m \right\rbrace \subseteq V$, let $U= \left\lbrace U_1, \dots, U_m \right\rbrace$ be the set of corresponding unit vectors. We will define a way to measure the closeness of the $U_i$ to each other using the metric $\Gamma$. To do this we define the diameter of $U$ as $$D = \displaystyle\max_{i,j} \left\lbrace \Gamma \left(U_i, U_j \right) \right\rbrace$$ If all of the vectors are taken in the standard way to be points on the unit sphere, then the diameter is a measure of the overall spread of the points. If the diameter is small then the vectors are all close together, hence highly correlated. Whereas if the diameter is large at least some of the points are far apart, hence not highly correlated. The benefit of the diameter is that it is an easy quantity to calculate, however it can be somewhat misleading. If, for instance, a large number of points are clustered together and there is one outlying point the diameter can be quite large despite the fact that the points are generally quite correlated.

We now proceed to generalize the correlation measure defined by $D$. Let $T$ be a collection of $t$ points on the $n$-sphere, and let $D$ be the set of $n$-simplices made up of points in $T$.

\begin{defn}\label{Mna}
$M_{n,a} \left( T \right) = \displaystyle\max_{\Delta \in D} \left\lbrace \text{Vol} \left( \Delta \right) \right\rbrace$
\end{defn}

This maximum is taken over the $C = C\left( T,n \right) = {t \choose n}$ different simplices made of points in $T$. 

\begin{defn}\label{Mn}
$M_{n} \left( T \right) = \text{Vol}\left(H\right)$. 
\end{defn}

The volume used in \ref{Mn} is the spherical volume and $H$ is the convex hull of the points of $T$ with respect to the spherical measure. That is, it is the smallest geodesically convex set containing $T$. (Geodesically convex means that any two points in the set have the minimal geodesic between them completely in the set as well.)

The volume is computed by constructing the convex hull of $T$, then disregarding all the points of $T$ not contributing to the hull. The hull is then divided into its `essential' $n$-simplices and the volumes of these simplices are summed.

$M_n$ and $M_{n,a}$ are each measures of $n$-dimensional volume. $M_{n,a}$ benefits from being easily computable. $M_n$, though harder to compute, gives a better measure of the overall spread of the vectors. However, in the one dimensional case we have that $M_1 = M_{1,a} = D$, the diameter. The reason for this is that when making the hull to compute $M_1$ all but the two furthermost points will be disregarded. This equality is not true in general, a fact which can be easily observed by plotting four points forming a quadrilateral where $M_{2,a} < M_2$. In the general case however we do have the inequality $M_{n,a} \left( T \right) \leq M_n \left( T \right)$. This follows since the maximal simplex will necessarily be a subset of the convex hull. Since volume is monotonic we have the inequality.

Assume that $s$ of the $t$ points of $T$ are essential to the convex hull. There is a constant $B = B\left( s,n \right)$ defining the number of essential simplices that compose the convex hull. i.e. $$M_n \left( T \right) = \text{ the sum of the volume of }B \text{-essential simplices}$$ Replacing the volume of each spherical simplex with the maximal one, that is $M_{n,a} \left( T \right)$, we get the following inequalities $$M_{n,a} \left( T \right) \leq M_n\left( T \right) \leq B \cdot  M_{n,a} \left( T \right)$$
Since $B$ depends only on the number of points in $T$ we see that, for a fixed data set, $M_n$ and $M_{n,a}$ differ by at most a fixed constant. 

To relate $M_n$ and $M_{n,a}$ to section \ref{TD} we note that when $T$ time dependent random variables are looked at over a summation window of length $K=n+1$ then we get $T$ points on the $n$-sphere. In this situation we can apply the measures of spread given by $M_n \left(T\right)$ or $M_{n,a}\left( T \right)$ or $M_{k,a}\left( T \right)$ for $k < n$.

\section{Topology of earth's climate indices and phase-locked states}\label{appendix}
In this section we apply our new correlation measure to data from Douglass's paper \cite{Douglass2010}. In \cite{Douglass2010} the diameter ($M_{1}=M_{1,a}$) is used to analyze a set of climate data, in this section we use $M_{2,a}$ to analyze the same data. Comparing the results of the new analysis to Douglass's original analysis shows the increased effectiveness of the new correlation measure.

Various regions of the Earth's climate system are characterized by temperature and
pressure indices. Douglass \cite{Douglass2010}, in a study of a global set of four indices, defines a distance
\begin{equation}
\hat{\Gamma}_{ij}\left( t \right) = \cos^{-1}\left( \vert \hat{\rho}\left[X_i \left( t \right), X_j \left( t \right) \right] \vert \right)
\end{equation}
between indices that satisfies the properties required to be a metric (\ref{metric}) where $\rho\left( X_i \left( t \right), X_j \left( t \right) \right)$ is the Pearson correlation coefficient. Note that the distance $\Gamma$ is an angle. 

In (\ref{CorrMeasures}) the correlation among a set of indices can be measured, using $M_{i,a}$ by taking the volumes of $i$-simplices. In \cite{Douglass2010}, Douglass uses the diameter of the metric space $\left(I_0, \Gamma \right)$, defined as 

\begin{equation}
D_{I_0} \left( t \right) = \displaystyle\max_{i,j} \left\lbrace \Gamma \left[X_i \left( t \right), X_j \left( t \right) \right] \vert i,j \in I_0 \right\rbrace 
\end{equation}

In the notation of (\ref{CorrMeasures}) $D = M_{1,a}$. Geometrically, $D$ selects the largest angle $\Gamma \left( X_i, X_j \right)$ among the set. The diameter $D$ may be considered a ``dissimilarity" index because large $D$ means weak correlation. Thus, the minima in $D$ are associated with high correlation among the elements of the set. In Douglass, \cite{Douglass2010}, two cases were considered: (1) the set of 3 Pacific ocean indices; (2) the global set of 4 indices (6 independent pairs). The $D$ of the global set is shown (in red) in Figure 1.

The maximal area $M_{2,a}$, the generalized correlation measure, was computed for the same four indices of \cite{Douglass2010}. The plot for the calculation is shown (blue) in Figure 1ab. Comparison of the two plots shows that the area measure reveals more minima (30) than the diameter (18). The various minima are indicated by arrows in Figure 1ab, and a list of dates is given in Table 1.
\newpage

Table 1: Date of various minima in plots of diameter $D = M_{1,a}$ or area $A= M_{2,a}$. (Minima are identified with a change in the phase-locked state of the Earth's climate system).
\begin{center}
\includegraphics[width=130mm]{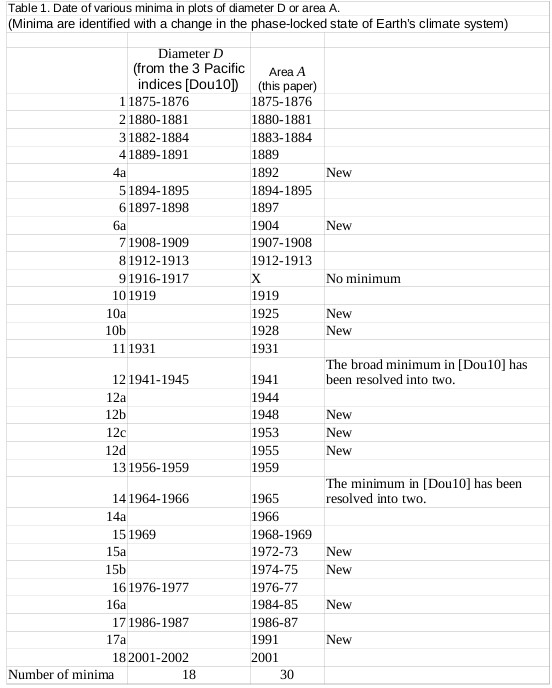}
\end{center}
\newpage
Figure 1(a) 1870-1940. (b)1940-2010
The plots are for two different correlation measures among a set of four global climate indices-- the diameter $D= M_{1,a}$ (in red) and the area $A=M_{2,a}$ (in blue) are defined in the text. Minima correspond to high correlation. The diameter $D$ plot shows 18 identified minima while the area $A$ plot shows 30. Comparisons are given in table 1.
\begin{center}
\includegraphics[width=130mm, height = 80mm]{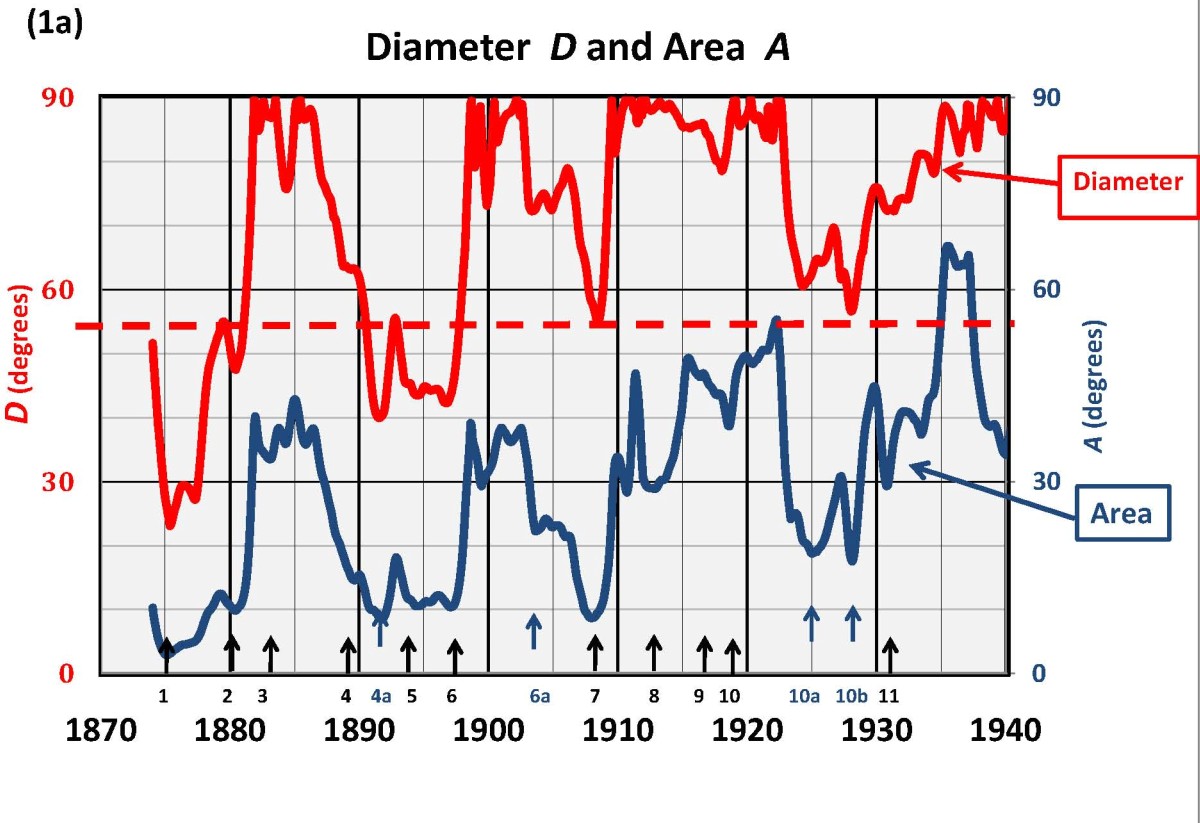} 
\includegraphics[width=130mm, height = 80mm]{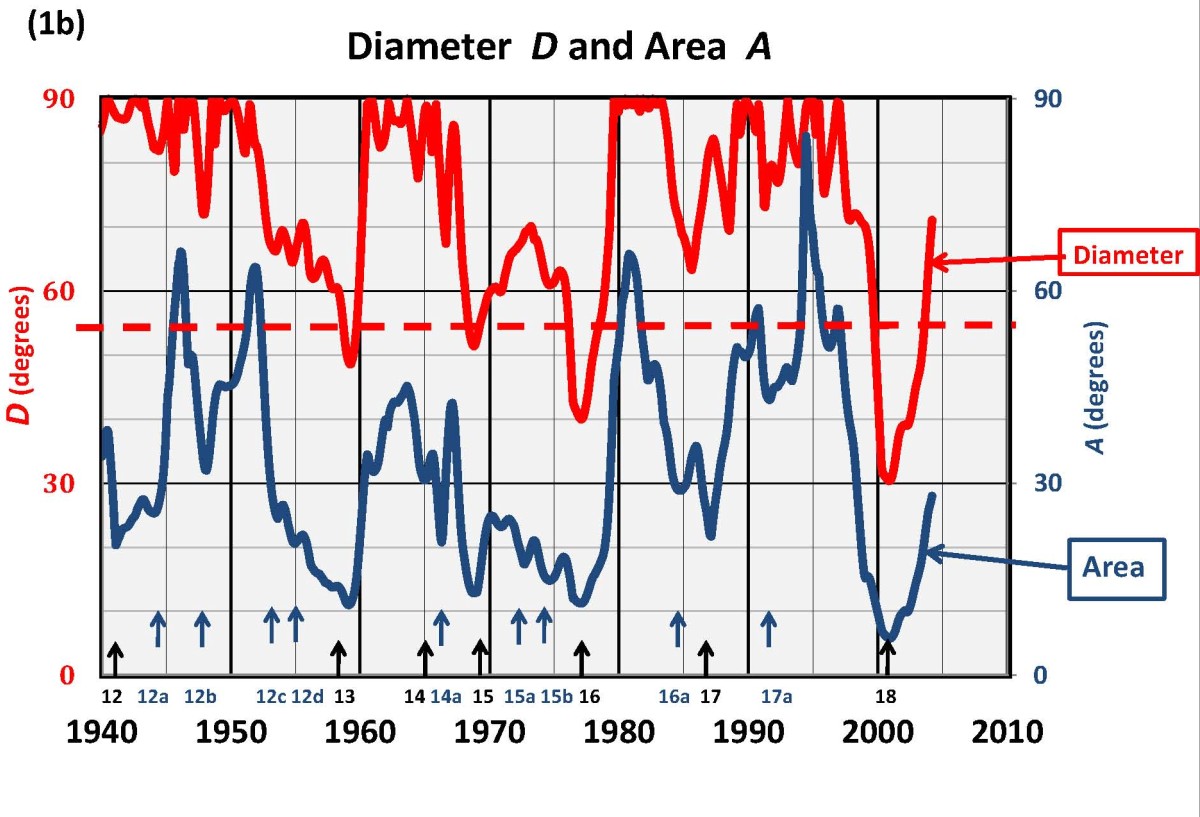}
\end{center}

\section{Summary}\label{Summary}
By using covariance on a set of time independent random variables or the covariance defined by the Pearson correlation on a set of time dependent variables we create metrics $\Gamma$ and $\hat{\Gamma}$ (resp.) on the unit sphere (resp. projective space) of the corresponding formal vector spaces. If $V$ is the $n$-dimensional formal vector space whose basis is the set of random variables $\left\lbrace X_1, \dots, X_n \right\rbrace$, we use $\Gamma$ or $\hat{\Gamma}$ to create $M_n$ or $M_{n,a}$, two measures of spread on values taken by the $X_i$. In section \ref{appendix} we give an explicit example of showing the use of $M_{2,a}$ on a global set of climate indices.

The two measures of spread, differ by at most a fixed multiplicative constant, so for theoretical purposes they are of equivalent use. However when applied they have can have different values. The volume of the convex hull created of $\left\lbrace X_1, \dots, X_n \right\rbrace$, given by $M_n$ is the most precise measure of the correlation of the $X_i$, however it is computationally difficult. The maximal volume of all possible $n$-simplices defined by the $X_i$, given by $M_{n,a}$, is a rougher measure of correlation. However $M_{n,a}$ is a simpler computation than $M_n$. 

In the 2-dimensional example, where all the vectors lie on the 2-sphere one can apply $M_{2.a}$, $M_2$, or $M_{1,a}= \text{Diameter}$. But in general $M_{1,a}$ is coarser than $M_{2,a}$ but is significantly easier to compute. For example, in \cite{Douglass2010} and section \ref{appendix} the use of $M_{2,a}$ yields much finer and cleaner results than the use of $M_{1,a}$. More generally in $n$-dimensions $M_{l}$ and $M_{l,a}$ for any $l \leq n$ and one sacrifices accuracy for ease.

\bibliographystyle{plain}
\bibliography{bib}
\end{document}